\documentclass[12pt,a4paper]{article}
\usepackage{amssymb}
\usepackage{amsmath}
\usepackage{amsfonts}
\usepackage[pdftex]{graphicx}
\graphicspath{{./graphics/}}
\DeclareGraphicsExtensions{.pdf,.png,.jpg}
\usepackage{amscd}
\usepackage{a4wide}
\usepackage{microtype}
\usepackage{bbm}
\usepackage{slashed}
\usepackage{fancyhdr}
\usepackage{amsthm}
\usepackage{mathrsfs}
\usepackage{hyperref}
\usepackage{color} 
\usepackage{pstricks}
\usepackage{enumerate}
\usepackage{epstopdf}

\definecolor{darkred}{rgb}{0.8,0.1,0.1}
\hypersetup{
     colorlinks=true,         
     linkcolor=darkred,
     citecolor=blue,
}

\usepackage{comment}
\usepackage[all,cmtip]{xy}

\usepackage{marvosym}

\usepackage{geometry}
\newtheorem{theorem}{\rmfamily\bfseries{Theorem}}[section]

\newtheorem{lemma}[theorem]{\rmfamily\bfseries{Lemma}} 
\newtheorem{proposition}[theorem]{\rmfamily\bfseries{Proposition}}

\theoremstyle{remark}

\newtheorem{remark}{Remark}

\def\un{1\kern-3pt \rm I}

\numberwithin{equation}{section}

\def\1{\mathbbm{1}}

\newcommand{\oN}{{\mathbb N}}
\newcommand{\oR}{{\mathbb R}}

\newcommand{\oC}{{\mathbb C}}

\title{\bf{On the stability and spectral properties of the two-dimensional Brown-Ravenhall operator
with a short-range potential}}

\author{Magno B. Alves,$^{*}$ Oswaldo M. Del Cima$^\dag$ and Daniel H.T. Franco$^{\dag \ddag}$ \vspace{4mm}\\
       $^{*}$ Universidade Federal de Juiz de Fora, Departamento de Matem\'atica,\\
      Campus Universit\'ario, Bairro Martelos,\\
      Juiz de Fora, MG, Brasil, CEP: 36.036-900.\vspace{4mm}\\
      $^\dag$ Grupo de F\'\i sica-Matem\'atica e Teoria Qu\^antica dos Campos,\\
      Universidade Federal de Vi\c cosa, Departamento de F\'\i sica,\\
      Av. Peter Henry Rolfs s/n, Campus Universit\'ario,\\
      Vi\c cosa, MG, Brasil, CEP: 36570-900.\vspace{4mm}\\
      $^\ddag$ Centro de Estudos de F\'\i sica Te\'orica,\\
      Setor de F\'\i sica-Matem\'atica,\\
      Rua Rio Grande do Norte 1053/302, Funcion\'arios,\\
      Belo Horizonte, Minas Gerais, Brazil, CEP 30130-131.\vspace{4mm}\\
{\small e-mail: \texttt{magno\underline{~}branco@yahoo.com.br, oswaldo.delcima@ufv.br, daniel.franco@ufv.br}}
}

\date{\today}

\begin{document}

\maketitle

\begin{abstract}
The Brown-Ravenhall operator was initially proposed as an alternative to describe the fermion-fermion interaction
via Coulomb potential and subject to relativity. This operator is defined in terms of the associated Dirac operator
and the projection onto the positive spectral subspace of the free Dirac operator. In this paper, we propose to analyze
a modified version of the Brown-Ravenhall operator in two-dimensions. More specifically, we consider the Brown-Ravenhall
operator with a short-range attractive potential given by a Bessel-Macdonald function (also known as $K_0$-potential)
using the Foldy-Wouthuysen unitary transformation. Initially, we prove that the two-dimensional
Brown-Ravenhall operator with $K_0$-potential is bounded from below when the coupling constant is below a specified
critical value (a property also referred to as stability). A major feature of this model is the fact that it does not
cease to be bounded below even if the coupling constant is above the specified critical value. We also investigate
the nature of the spectrum of this operator, in particular the location of the essential spectrum, and the existence
of eigenvalues, which are either isolated from the essential spectrum or embedded in it.
\end{abstract}




\section{Introduction}
\label{Sec1}
\hspace*{\parindent}
The quantum electrodynamics in three space-time dimensions (QED$_3$) has been drawn attention, 
since the works by Schonfeld, Deser, Jackiw and Templeton~\cite{schonfeld,DJT}, 
as a potential theoretical framework to be applied to quasi-planar condensed matter 
systems~\cite{MWH}, namely high-$T_{\rm c}$ superconductors~\cite{high-Tc,CWFH}, quantum Hall~\cite{QHE}, 
topological insulators~\cite{topological-insulators}, topological superconductors~\cite{TSCond} 
and graphene~\cite{graphene,graphene1,graphene2,WeWaEm}. Thenceforth,  planar quantum electrodynamics models 
have been studied in many physical configurations: small (perturbative) and large (non perturbative) gauge 
transformations, abelian and non-abelian gauge groups, fermions families, even or odd under parity, compact 
space-times, space-times with boundaries, curved space-times, discrete (lattice) space-times, external fields 
and finite temperatures. In condensed matter systems, quasiparticles usually stem from two-particle (Cooper pairs), 
particle-quasiparticle (excitons) or two-quasiparticle (bipolarons) non-relativistic bound states. 

Regarding the present article and the physics of new ``Dirac materials''\cite{Wehling}, in particular due to the importance
of the Dirac's equation in the description of graphene, together with the fact that there are QED$_3$ models in which,
fermion-fermion, fermion-antifermion or antifermion-antifermion scattering potentials -- mediated by massive scalars or
vector mesons -- can be attractive\footnote{While the obtained scattering potentials for $p$-wave states showed up
repulsives, the $s$-wave states (angular momentum state $\ell=0$) show attractive.} and of $K_0$-type (a Bessel-Macdonald
function)~\cite{CWFH,MWH,graphene2,WeWaEm,Wado1,Emer,Well,Laz,WODC}, here we propose to discuss quantum
relativistic effects by a short-range potential to treat the interaction between fermion-fermion, fermion-antifermion or antifermion-antifermion
in $d=2$, without magnetic field. We deal with the potential theory in spaces of Bessel potentials. More specifically, in this
article we study the two-dimensional Brown-Ravenhall operator, {\em i.e.}, the projection of the Dirac operator perturbed by
a $K_0$-potential (the two-dimensional Brown-Ravenhall operator in $d=2$ perturbed by a Coulomb potential has been
analyzed by Bouzouina~\cite{Bouzo} and Walter~\cite{Wal}). Our main results are as follows: in Section \ref{Sec3}, we prove
that the Brown-Ravenhall operator with $K_0$-potential in $d=2$ is bounded from below when the coupling constant is below
a specified critical value (a property also referred to as stability). We do this following two paths: through relativistic Sobolev
inequality and relativistic Hardy inequality (also referred to as Kato inequality). We provide at the end of the article an
Appendix A, where we address the existence of an inequalitiy involving the Hardy and Sobolev relativistic inequalities that 
we will call Hardy-Sobolev-Maz'ya type inequality (or HSM-type inequality for brevity). We also show that the stability of the
Brown-Ravenhall operator with $K_0$-potential in $d=2$ implies his self-adjointness. Closing this part of the article, we show
that the two-dimensional Brown-Ravenhall operator with $K_0$-potential does not cease to be bounded below even if the
coupling constant is above the specified critical value, showing that a complete implosion will never occur for two relativistic
fermions interacting via an attractive potential of the Bessel-Macdonald type, unlike to usual Brown-Ravenhall operator with
coulombian potential. The nature of the spectrum of this operator in turn is addressed in Section \ref{Sec4}, in particular the
location of the essential spectrum, and the existence of eigenvalues, which are either isolated from the essential spectrum
or embedded in it. The analysis of embedded eigenvalues is based on a simple abstract virial theorem. Due to the purpose
of this paper, for the convenience of the reader, in Appendix B we gather a few facts about Bessel's potential class, of which
the potential $K_0$ is part.

\begin{remark}
The following notations will be used consistently throughout the article: $\boldsymbol{x},\boldsymbol{y},\boldsymbol{z},\ldots$
will denote points of the $d$-dimensional euclidean space $\oR^d$, $|\boldsymbol{x}-\boldsymbol{y}|$ the euclidean
distance of the points $\boldsymbol{x},\boldsymbol{y}$, $|\boldsymbol{x}|=|\boldsymbol{x}-\boldsymbol{0}|$,
$\boldsymbol{p},\boldsymbol{q},\ldots$ points of the dual space, $\boldsymbol{p} \cdot \boldsymbol{x}$ the inner
product of the vectors $\boldsymbol{p}$ and $\boldsymbol{x}$. The gradient $\boldsymbol{\nabla} \Psi$ of a differentiable
function $\Psi$ is $\boldsymbol{\nabla} \Psi=(\partial \Psi/\partial x_1,\ldots,\partial \Psi/\partial x_d)$. 
If $\Psi,\Phi \in L_2(\oR^d)$, then we set $\langle \Psi,\Phi \rangle=\int_{\oR^d} \overline{\Psi}(x) \Phi(x)\,dx$. $\Psi * \Phi$
will denote the convolution of $\Psi$ and $\Phi$, $\widehat{\Psi}$ the Fourier transform of $\Psi$. We are adopting the
following convention for the Fourier transform on $\oR^d$:
\begin{align}
\bigl[{\mathscr F}\Psi\bigr](\boldsymbol{p})&=\widehat{\Psi}(\boldsymbol{p})
=\int_{\oR^d} e^{i \hbar^{-1}\boldsymbol{p} \cdot \boldsymbol{x}}\,{\Psi}(\boldsymbol{x})
\,d\boldsymbol{x}\,\,,
\label{FouTrans}
\end{align}
\begin{align}
\bigl[{\mathscr F}^{-1}\widehat{\Psi}\bigr](\boldsymbol{x})&=\Psi(\boldsymbol{x})
=\frac{1}{(2\pi \hbar)^{d}}\int_{\oR^d} e^{-i \hbar^{-1}\boldsymbol{p} \cdot \boldsymbol{x}}\,\widehat{\Psi}(\boldsymbol{p})
\,d\boldsymbol{p}\,\,,
\label{InvFouTrans}
\end{align}
with $\Psi \in {\mathscr S}(\oR^d)$. Here, by ${\mathscr S}(\oR^d)$ we mean the set of all rapidly decreasing functions
on $\oR^d$. The $\hbar$ is introduced to keep the units consistent with the physical interpretation.
Of course, the invariance of space ${\mathscr S}(\oR^d)$ under the Fourier transform implies that
$\widehat{\Psi} \in {\mathscr S}(\oR^d)$.
\end{remark}

\section{The Dirac-Bessel-Macdonald operator restricted to its positive spectral subspace in $d=2$}
\label{Sec2}
\hspace*{\parindent}
The Brown-Ravenhall operator~\cite{BR} was initially proposed as an alternative to describe the fermion-fermion interaction
via Coulomb potential and subject to relativity. This operator is defined in terms of the associated Dirac operator
and the projection onto the positive spectral subspace of the free Dirac operator. In what follows, we consider a version
of the two-dimensional Brown-Ravenhall operator with the $K_0$-potential (a Bessel-Macdonald function)
\begin{align}
\boldsymbol{\cal B}(\boldsymbol{x})
=\Lambda_+\bigl(\boldsymbol{\cal D}_0(\boldsymbol{x})-\gamma K_0(\beta |\boldsymbol{x}|)\bigr)\Lambda_+\,\,,
\label{OpBR}
\end{align}
where $\gamma$ is the coupling parameter taken to be non-negative. The constants $\gamma$
and $\beta$ shall depend on some model parameters, like coupling constants, characteristic lengths, mass parameters or
vacuum expectation value of a scalar field. From expression $V(\boldsymbol{x})=-\gamma K_0(\beta |\boldsymbol{x}|)$ we
see that $\gamma$ has energy dimension and gives us an energy scale for the interaction among the two particles. In turn,
the parameter $\beta$ has inverse length dimension, thus fixing a length scale, an interaction range, which is related to
the mass of the boson-mediated quantum exchanged during the two particle
scattering~\cite{MWH,CWFH,graphene2,Wado1,Emer,Well,Laz,WODC}

In (\ref{OpBR}) the notation is as follows:
\begin{enumerate}
\item The operator $\boldsymbol{\cal D}_0$ is the {\em free Dirac operator} in $d=2$; it is a first order operator acting on {\em spinor}-valued functions
$\Psi(\boldsymbol{x})=(\psi_1(\boldsymbol{x}),\psi_2(\boldsymbol{x}))$, with $2$ components, of the space variable
$\boldsymbol{x}=(x_1,x_2)$. We denote by $\oC^2$ the $2$-dimensional complex vector space in which the values
of $\Psi(\boldsymbol{x})$ lie. $\boldsymbol{\cal D}_0$ has the form
\begin{align*}
\boldsymbol{\cal D}_0=-i \hbar c\, \boldsymbol{\sigma} \cdot \boldsymbol{\nabla}+m c^2 \sigma_3
=-i \hbar c \left(\sigma_1 \frac{\partial}{\partial x_1}+\sigma_2 \frac{\partial}{\partial x_2}\right)
+mc^2 \sigma_3\,\,.
\end{align*}
where $\hbar$ is the Planck constant, $m > 0$ is the mass of the fermionic particle under consideration, $c$ is the velocity
of light and $\boldsymbol{\sigma}=(\sigma_1;\sigma_2)$ and $\sigma_3$ are the Pauli $2 \times 2$-matrices
\begin{align*}
\sigma_1=
\begin{pmatrix}
0 & 1 \\[3mm]
1 & 0
\end{pmatrix}\,\,,\quad
\sigma_2=
\begin{pmatrix}
0 & -i \\[3mm]
i & 0
\end{pmatrix}\,\,,\quad
\sigma_3=
\begin{pmatrix}
1 & 0 \\[3mm]
0 & -1
\end{pmatrix}\,\,.
\end{align*}
The $\sigma_j$ matrices are introduced in view of making the Dirac operator a square root of the Laplace operator;
they satisfy by construction the following anti-commutating relations:
\[
\sigma_j \sigma_k+\sigma_k \sigma_j=2 \delta_{jk} {\un}_{2 \times 2}\,\,,
\quad j,k=1,2\,\,.
\]

\begin{remark}
The free Dirac operator $\boldsymbol{\cal D}_0$ is essentially self-adjoint on the dense subspace
$C_0^\infty(\oR^2 \setminus \{0\};\oC^2)$ and self-adjoint on the Sobolev space
$\boldsymbol{\mathfrak{Dom}}(\boldsymbol{\cal D}_0)=H^1(\oR^2;\oC^2)$, its spectrum is given by
\[
\sigma(\boldsymbol{\cal D}_0)=(-\infty,-mc^2] \cup [mc^2,+\infty)\,\,,
\]
and it has as form domain the space $\boldsymbol{\mathfrak{Q}}(\boldsymbol{\cal D}_0)=H^{1/2}(\oR^2;\oC^2)$
(see for example~\cite[Chapter 7]{LM} for more details on the spaces $H^1(\oR^2;\oC^2)$ and $H^{1/2}(\oR^2;\oC^2)$).
Naturally, the negative spectrum is associated with antiparticles, in relativistic theories.
\end{remark}

\item $\Lambda_+ \overset{\rm def.}{=} \chi_{(0,\infty)}\bigl(\boldsymbol{\cal D}_0\bigr)$, where $\chi_{(0,\infty)}$
is the characteristic function of the interval $(0,+\infty)$, denotes the projection of $L_2(\oR^2;\oC^2)$ onto the
positive spectral subspace of $\boldsymbol{\cal D}_0$; namely,
\begin{align*}
\Lambda_+=\frac{1}{2}
\left({\un}_{2 \times 2}
+\frac{-i \hbar c\,\boldsymbol{\sigma} \cdot \boldsymbol{\nabla}+m c^2 \sigma_3}
{\sqrt{-\hbar^2 c^2\,\Delta+m^2 c^4}}\right)\,\,,
\end{align*}
where $\Delta$ is the laplacian operator on $\oR^2$. Note that 
$\boldsymbol{\cal D}_0 \Lambda_+=\Lambda_+ \boldsymbol{\cal D}_0=\sqrt{-\hbar^2 c^2\,\Delta+m^2 c^4}\;\Lambda_+$.
The last equality is a consequence of the following fact: in Fourier variables, the projector $\Lambda_+$ is a multiplication
operator given by the following expression:
\begin{align*}
\widehat{\Lambda}_+
=\frac{1}{2}
\left({\un}_{2 \times 2}
+\frac{c\,\boldsymbol{\sigma} \cdot \boldsymbol{p}+m c^2 \sigma_3}
{\sqrt{c^2\,|\boldsymbol{p}|^2+m^2 c^4}}\right)\,\,.
\end{align*}

Hence, the Dirac-Bessel-Macdonald operator restricted to its positive spectral subspace, or
the Brown-Ravenhall operator with the $K_0$-potential, is given formally as
\begin{align*}
\boldsymbol{\cal B}(\boldsymbol{x})
=\Lambda_+ \sqrt{-\hbar^2 c^2\,\Delta+m^2 c^4}\;\Lambda_+
-\gamma\;\Lambda_+ K_0(\beta |\boldsymbol{x}|) \Lambda_+\,\,,
\end{align*}
acting in $L_2(\oR^2;\oC^2)$, or, equivalently
\begin{align*}
\boldsymbol{\cal B}(\boldsymbol{x})
=\Lambda_+ \sqrt{-\hbar^2 c^2\,\Delta+m^2 c^4}
-\gamma\;\Lambda_+ K_0(\beta |\boldsymbol{x}|)\,\,,
\end{align*}
acting in ${\mathscr H}_+ \overset{\rm def.}{=} \Lambda_+\bigl(L_2(\oR^2;\oC^2)\bigr)$.
\end{enumerate}

\section{Boundedness from below and all that}
\label{Sec3}
\hspace*{\parindent}
In applications it is often very important to determine the lowest point of the spectrum of a self-adjoint operator.
This problem makes sense only if the operator is bounded from below, since otherwise the spectrum extends to
$-\infty$. In this section, we particularly focus on the problem of the boundedness from below of the Brown-Ravenhall
operator with $K_0$-potential, which is related to the stability of relativistic systems in two dimensions. More precisely,
we are interested in finding the largest value of $\gamma$ such that the quadratic form $\langle \Psi,\boldsymbol{\cal B} \Psi \rangle$ is
bounded from below. In effect, we reach this bounded from below through two paths: through relativistic Sobolev inequality
and relativistic Hardy inequality (also referred to as Kato inequality).

\subsection{Preamble: reduction of spinors}
\label{Subsec3a}
\hspace*{\parindent}
The first step in order to prove the boundedness from below of the Brown-Ravenhall operator with $K_0$-potential is
a reduction of spinors. We will follow the same strategy as Zelati-Nolasco~\cite{ZeNo}: we now introduce the
{\em Foldy-Wouthuysen transformation} (FW), given by a unitary transformation $U_{\rm FW}$ which transforms the
free Dirac operator into the diagonal form (see details in \cite[the case in $d=1+2$]{Binegar} and \cite[the case in 
$d=1+3$]{Thaller})
\begin{align*}
\boldsymbol{\cal D}_{\rm FW}=U_{\rm FW} \boldsymbol{\cal D}_0 U_{\rm FW}^{-1}
=\begin{pmatrix}
\boldsymbol{\cal H}_0 & 0 \\[3mm]
0 & -\boldsymbol{\cal H}_0
\end{pmatrix}
=\sigma_3 \boldsymbol{\cal H}_0\,\,,
\end{align*}
where $\boldsymbol{\cal H}_0 \overset{\rm def.}{=}\sqrt{-\hbar^2 c^2\,\Delta+m^2 c^4}$ is the so-called quasi-relativistic operator
(the relativistic (free) hamiltonian operator). This operator has been extensively studied for a long time (we refer
to~\cite{Weder1,Weder2,Herbst,DaubLieb,Daub,Nard,Umeda,Raynal}).  

\begin{remark}
With the usual quantization rule $\boldsymbol{p} \mapsto -i \hbar \boldsymbol{\nabla}$, let us recall that to the operator
$\boldsymbol{\cal H}_0$ can be defined for all $\Psi \in H^1(\oR^2;\oC^2)$ as the inverse Fourier transform
of the $L_2$-function $\sqrt{c^2 |\boldsymbol{p}|^2+m^2 c^4}~\widehat{\Psi}(\boldsymbol{p})$ (where $\widehat{\Psi}$
denotes the Fourier transform of $\Psi$). To $\boldsymbol{\cal H}_0$ we can associate the following quadratic form
\begin{align*}
q_{_H}(\Phi,\Psi)
\overset{\rm def.}{=} \langle \Phi,\boldsymbol{\cal H}_0 \Psi \rangle=\frac{1}{(2\pi \hbar)^{2}} \int_{\oR^2}
\sqrt{c^2 |\boldsymbol{p}|^2+m^2 c^4}~\overline{\widehat{\Phi}}(\boldsymbol{p}) \widehat{\Psi}(\boldsymbol{p})~d\boldsymbol{p}\,\,,
\end{align*}
which can be extended to all functions $\Phi,\Psi \in \boldsymbol{\mathfrak{Q}}(\boldsymbol{\cal H}_0)=H^{1/2}(\oR^2;\oC^2)$, where
\begin{align*}
H^{1/2}(\oR^2;\oC^2)=\left\{\Psi \in L_2(\oR^2;\oC^2) \mid 
\int_{\oR^2} (1+|\boldsymbol{p}|^2)^{1/2}~|\widehat{\Psi}(\boldsymbol{p})|^2~d\boldsymbol{p} < \infty\right\}\,\,.
\end{align*}
It is known that $\boldsymbol{\cal H}_0$ restricted on $C_0^\infty(\oR^d)$ is essentially self-adjoint, and that
$\sigma(\boldsymbol{\cal H}_0)=\sigma_{\rm ess}(\boldsymbol{\cal H}_0)=[mc^2,\infty)$. An excellent mathematical, comprehensive
and self-contained analysis of the spectral properties of the operators $\boldsymbol{\cal B},\boldsymbol{\cal D}_0$ and $\boldsymbol{\cal H}_0$
(perturbed by the Coulomb potential) can be found in~\cite{BaEv}.
\label{Note01}
\end{remark}

Under the FW-transformation the projector $\Lambda_+$ becomes simply
\begin{align*}
\Lambda_{+{\rm FW}} \overset{\rm def.}{=} U_{\rm FW} \Lambda_+ U_{\rm FW}^{-1}
=\frac{1}{2}
\left({\un}_{2 \times 2}+\sigma_3\right)\,\,.
\end{align*}
Therefore the positive energy subspace for $\boldsymbol{\cal D}_{\rm FW}$ is simply given by
\begin{align*}
{\mathscr H}_+=\left\{\Psi=\binom{\psi}{0} \in L_2(\oR^2;\oC^2) \mid \psi \in L_2(\oR^2;\oC) \right\}
\end{align*}

In the FW-representation the associated quadratic form acting on ${\mathscr H}_+$ is defined by
\begin{align}
\langle \varphi,\boldsymbol{\cal B}_{\rm FW} \psi \rangle_{L_2(\oR^2;\oC)}=
\left\langle \varphi,\boldsymbol{\cal H}_0 \psi \right\rangle_{L_2(\oR^2;\oC)}
+ \langle \varphi,V_{\rm FW} \psi \rangle_{L_2(\oR^2;\oC)}\,\,,
\label{OpBRFW}
\end{align}
for any $\varphi,\psi \in H^{1/2}(\oR^2;\oC)$, where $V_{\rm FW} \psi=Q^*U_{\rm FW}VU_{\rm FW}^{-1}Q\psi$, with
$V(\boldsymbol{x})=-\gamma K_0(\beta |\boldsymbol{x}|)$ and
\begin{align*}
&Q:  \oC \to \oC^2\,\,,\quad Q(z_1)=(z_1,0)\,\,, \\[3mm]
&Q^*:\oC^2 \to \oC\,\,,\quad Q^*(z_1,z_2)=z_1\,\,,
\end{align*}
so that
\begin{align*}
\left\langle \varphi,\boldsymbol{\cal H}_0 \psi \right\rangle_{L_2(\oR^2;\oC)}
&=\left\langle \Lambda_+U_{\rm FW}^{-1}Q\varphi,\boldsymbol{\cal D}_0 \Lambda_+U_{\rm FW}^{-1}Q\psi \right\rangle_{L_2(\oR^2;\oC^2)} \\[3mm]
&=\left\langle \Lambda_+U_{\rm FW}^{-1} \binom{\varphi}{0},\boldsymbol{\cal D}_0 \Lambda_+U_{\rm FW}^{-1} \binom{\psi}{0}\right\rangle_{L_2(\oR^2;\oC^2)}\,\,,
\end{align*}
and
\begin{align*}
\left\langle \varphi,V_{\rm FW}\psi \right\rangle_{L_2(\oR^2;\oC)}
&=\left\langle U_{\rm FW}^{-1}Q\varphi,V U_{\rm FW}^{-1}Q\psi \right\rangle_{L_2(\oR^2;\oC^2)} \\[3mm]
&=\left\langle \Lambda_+U_{\rm FW}^{-1} \binom{\varphi}{0},V\Lambda_+U_{\rm FW}^{-1} \binom{\psi}{0}\right\rangle_{L_2(\oR^2;\oC^2)}\,\,.
\end{align*}
Note that $U_{\rm FW}^{-1}Q\varphi=\Lambda_+U_{\rm FW}^{-1}Q\varphi \in \Lambda_+\bigl(L_2(\oR^2);\oC^2\bigr)$
for any $\varphi \in L_2(\oR^2;\oC)$.

Using the above description, for any $\psi$ in the positive spectral subspace, the expectation
of $\boldsymbol{\cal B}$ in the state $\Psi$ in the FW-representation is associated to the quadratic form
\begin{align}
\langle \psi,\boldsymbol{\cal B}_{\rm FW} \psi \rangle=
\left\langle \psi,\boldsymbol{\cal H}_0 \psi \right\rangle
-\gamma \langle \psi,K_0(\beta |\boldsymbol{x}|) \psi \rangle\,\,.
\label{OpBRFWa}
\end{align}

Hence, the transition from $\Psi \in L_2(\oR^2;\oC^2)$ to the reduced spinor $\psi \in L_2(\oR^2;\oC)$ through the
introduction of the operator $\boldsymbol{\cal B}_{\rm FW}$ is possible because we are working in ${\mathscr H}_+$.
The quadratic form (\ref{OpBRFWa}) defines a self-adjoint operator $\boldsymbol{\cal B}_{\rm FW}$ if we can show that the form
$\langle \psi,\boldsymbol{\cal B}_{\rm FW} \psi \rangle$ is bounded from below. Of course, $\langle \Psi,\boldsymbol{\cal B} \Psi \rangle$
is bounded from below if and only if $\langle \psi,\boldsymbol{\cal B}_{\rm FW} \psi \rangle$ is bounded from below. In that case,
Eq.(\ref{OpBRFWa}) will define a self-adjoint operator $\boldsymbol{\cal B}_{\rm FW}$.

\begin{remark}
The $K_0$-potential satisfies the following properties: $\bf (P1)$ $K_0 \in L_2(\oR^2)+L_\infty(\oR^2)$
and $\bf (P2)$ there exists a constant $C \in (0,1)$ such that $|\gamma \langle \psi,K_0 \psi \rangle| \leqslant C \langle \psi,\boldsymbol{\cal H}_0 \psi \rangle$,
for all $\psi \in H^{1/2}(\oR^2;\oC)$. The validity of $\bf (P1)$ is proven in Lemma \ref{SpecLemma1}, while $\bf (P2)$ 
follows from important known inequalities: the relativistic Sobolev inequality and the relativistic Hardy inequality, respectively.
\end{remark}

\subsection{Boundedness from below via relativistic Sobolev inequality}
\label{Subsec3aa}
\hspace*{\parindent}
Our first proof of the boundedness from below of the Brown-Ravenhall operator with $K_0$-potential is based 
on the following result (see~\cite[Theorem 2.1]{CoTa} and \cite[Theorem 1.1]{CoTav}):

\begin{theorem}
For $d > 2s$, let $\psi \in H^s(\oR^d)$ and $q=2d/(d-2s)$. Then the following inequality holds:
\begin{align}
\|\psi\|_q^2 \leqslant {\cal S}_{d,s} \|(-\Delta)^{s/2} \psi\|_2^2= {\cal S}_{d,s} \left\langle \psi,(-\Delta)^s \psi \right\rangle\,\,,
\label{CoTa1}
\end{align}
where
\begin{align}
{\cal S}_{d,s}=2^{-2s} \pi^{-s} \frac{\Gamma\left(\frac{d-2s}{2}\right)}{\Gamma\left(\frac{d+2s}{2}\right)}
\left(\frac{\Gamma(d)}{\Gamma(d/2)}\right)^{2s/d}\,\,.
\label{CoTa2}
\end{align}
There is equality in (\ref{CoTa1}) if, and only if, $\psi(x)=A\bigl((x-x_0)^2+\eta^2)\bigr)^{-(n-2)/2}$, where $A \in \oR$,
$\eta > 0$ are fixed constants and with $x_0 \in \oR^n$.
\label{RSItheo}
\end{theorem}

\begin{remark}
The relativistic Sobolev inequality~\cite[Theorem 8.3]{LM} and \cite[Theorem 1.5]{Seri} is a particular case of inequality
(\ref{CoTa1}) when $s=1/2$, and for $s=1/2$ the best constant of Theorem \ref{RSItheo} is given in~\cite{LM}.
\end{remark}

\begin{remark}
The idea that every problem in physics can and should be analyzed in terms of physical dimensions
seems to have been missed by many studies like the one covered in this article. For this reason,
from this point on, both in the analysis of the boundedness from below via relativistic Sobolev inequality,
and the subsequent analysis via relativistic Hardy inequality, we will take into account the physical dimensions.
Here, we are adopting the {\em International System of Units}, since we are working explicitly with $\hbar$ and $c$.
\end{remark}

A well-known result (see~\cite[Theorem 1.66]{BCD}) asserts that $H^s(\oR^d)$ is continuously embedded into
$L_q(\oR^d)$ if $0 \leqslant s < d/2$ and $2 \leqslant q \leqslant 2d/(d-2s)$. In particular, for $s=1/2$ and $d=2$ this
reads as
\begin{align}
H^{1/2}(\oR^2) \hookrightarrow L_q(\oR^2)
\quad \text{if} \quad 2 \leqslant q \leqslant 4\,\,.
\label{Mag1}
\end{align}

Now, let $\psi$ be an arbitrary element in $H^{1/2}(\oR^2)$ and let $K_0(\beta |\boldsymbol{x}|)$ be the 
Bessel-MacDonald potential in $\oR^2$. It is well known that
\begin{align}
K_0(\beta |\boldsymbol{x}|) \in L_1(\oR^2) \cap L_2(\oR^2)\,\,.
\label{Mag2}
\end{align}

Since $\psi \in L_4(\oR^2)$, we have that
\begin{align}
|\psi|^2 \in L_2(\oR^2)\,\,.
\label{Mag3}
\end{align}

From (\ref{Mag2}), (\ref{Mag3}) and H\"older Inequality, we get
\begin{align}
K_0 |\psi|^2 \in L_1(\oR^2)\,\,.
\label{Mag4}
\end{align}
In particular, from (\ref{Mag4}) and Cauchy-Schwarz-Bunjakowski inequality, we have that
\begin{align*}
\langle \psi,K_0(\beta |\boldsymbol{x}|) \psi \rangle
&=\int_{\oR^2} \overline{\psi}(\boldsymbol{x}) K_0(\beta |\boldsymbol{x}|) \psi(\boldsymbol{x})\,d\boldsymbol{x} \\[3mm]
&=\int_{\oR^2} K_0(\beta |\boldsymbol{x}|) |\psi(\boldsymbol{x})|^2\,d\boldsymbol{x} \\[3mm]
&=\int_{\oR^2} \overline{K_0}(\beta |\boldsymbol{x}|) |\psi(\boldsymbol{x})|^2\,d\boldsymbol{x} 
\quad (K_0 \in \oR) \\[3mm]
&=\langle K_0,|\psi|^2 \rangle \\[3mm]
&\leqslant \|K_0\|_2 \||\psi|^2\|_2\,\,.
\end{align*}
Note that
\begin{align*}
\||\psi|^2\|_2=\left(\int_{\oR^2} (|\psi(\boldsymbol{x})|^2)^2\,d\boldsymbol{x}\right)^{1/2}
=\left(\int_{\oR^2} |\psi(\boldsymbol{x})|^4\,d\boldsymbol{x}\right)^{1/2} 
=\|\psi\|_4^2\,\,,
\end{align*}

From (\ref{Mag1}), we have that $\psi \in L_q(\oR^2)$, for $2 \leqslant q \leqslant 4$ since we are taking $\psi \in H^{1/2}(\oR^2)$.
Thus, we can conclude that
\begin{align*}
\langle \psi,K_0(\beta |\boldsymbol{x}|) \psi \rangle
\leqslant \|K_0\|_2 \|\psi\|_4^2\,\,,
\end{align*}
and, for all $\psi \in H^{1/2}(\oR^2;\oC)$. Accordingly, in the international system of units, the relativistic Sobolev inequality 
for $c |\boldsymbol{p}|$~\cite[Theorem 8.4]{LM} takes the form
\begin{align}
\|\psi\|_4^2 \leqslant \frac{1}{\hbar c} {\cal S}_{2,1/2} \left\langle \psi,\sqrt{-\hbar^2 c^2 \Delta}\;\psi \right\rangle\,\,,
\label{RSI}
\end{align}
where, according to Eq.(\ref{CoTa2}),
\begin{align*}
{\cal S}_{2,1/2}=2^{-1} \pi^{-1/2} \frac{\Gamma\left(\frac{1}{2}\right)}{\Gamma\left(\frac{3}{2}\right)}
\left(\frac{\Gamma(2)}{\Gamma(1)}\right)^{1/2}=\pi^{-1/2}
\quad \text{(by the particular values of the gamma function)}\,\,.
\end{align*}
Moreover, according the Table of Integrals of Gradshtein-Ryzhik~\cite[{\bf 6.521}, $6.^*$, p.665]{GR}, in polar coordinates,
\begin{align}
\|K_0\|_2=\left(2 \pi \int_{0}^\infty K^2_{0}(\beta r)\,r\,dr\right)^{1/2}=\frac{\pi^{1/2}}{\beta}\,\,.
\label{Grads}
\end{align}

Therefore, we have that
\begin{align}
\langle \psi,K_0(\beta |\boldsymbol{x}|) \psi \rangle
&\leqslant \frac{1}{\hbar c \beta}
\left\langle \psi,\sqrt{-\hbar^2 c^2 \Delta}\;\psi \right\rangle \nonumber \\[3mm]
&\leqslant \frac{1}{\hbar c \beta}
\left\langle \psi,\sqrt{-\hbar^2 c^2 \Delta+m^2 c^4}\;\psi \right\rangle\,\,.
\label{RSIa}
\end{align}

The whole bounded from below problem of the operator $\boldsymbol{\cal B}_{\rm FW}$ is captured in the estimate 
\begin{align*}
\left\langle \psi,\sqrt{-\hbar^2 c^2\,\Delta+m^2 c^4}\;\psi \right\rangle
-\gamma \langle \psi,K_0(\beta |\boldsymbol{x}|) \psi \rangle \geqslant 0\,\,,
\quad \forall\,\psi \in H^{1/2}(\oR^2;\oC)\,\,,
\end{align*}
and leads immediately to the question for what values of $\gamma$ such an estimate holds. Consequently, the inequality
(\ref{RSIa}) shows that operator $\boldsymbol{\cal B}_{\rm FW}$ is bounded from below if
\begin{align*}
\left(1-\gamma \frac{1}{\hbar c \beta}\right) 
\left\langle \psi,\sqrt{-\hbar^2 c^2\,\Delta+m^2 c^4}\;\psi \right\rangle \geqslant 0\,\,,
\end{align*}
In other words, $\langle \psi,\boldsymbol{\cal B}_{\rm FW} \psi \rangle$ is lower bounded if
\begin{align*}
\gamma \leqslant \gamma_{\rm c}^{\cal S}
=\hbar c \beta\,\,.
\end{align*}
This ends the first proof of the boundedness from below of the Brown-Ravenhall operator with $K_0$-potential.

\subsection{Boundedness from below via relativistic Hardy inequality}
\label{Subsec3b}
\hspace*{\parindent}
Our second proof of the boundedness from below of the Brown-Ravenhall operator with $K_0$-potential is based 
on the following result~\cite[Lemma 8.2]{LS}:

\begin{theorem}
Let $d \geqslant 2$, and let $\psi$ be a function in $H^{1/2}(\oR^d)$. Then there is the strict inequality
\begin{align}
\int_{\oR^d} \frac{|\psi(\boldsymbol{x})|^2}{|\boldsymbol{x}|}\,d\boldsymbol{x}
< {\cal H}_d^2 \int_{\oR^d} |\boldsymbol{p}| |\widehat{\psi}(\boldsymbol{p})|^2\,d\boldsymbol{p}
={\cal H}_d^2 \int_{\oR^d} \overline{\psi}(\boldsymbol{x}) \sqrt{-\Delta}\;\psi(\boldsymbol{x})\,d\boldsymbol{x}\,\,,
\label{LS1}
\end{align}
where the best possible value of the constant ${\cal H}_d$ is
\begin{align*}
{\cal H}_d=\frac{\Gamma\left(\frac{d-1}{4}\right)}{\sqrt{2}\,\Gamma\left(\frac{d+1}{4}\right)}\,\,.
\end{align*}
The equality is only attained if $\psi=0$, {i.e.}, for any bigger constant the inequality fails for some function
in $H^{1/2}(\oR^d)$.
\label{BaEvLe1}
\end{theorem}

\begin{remark}
This inequality goes back to Kato~\cite[Eq.(V.5.33)]{Kato} and Herbst~\cite[Theorem 2.5]{Herbst}.
See also~\cite[Theorem 1.7.1]{BaEvLe}.
\end{remark}

To apply Theorem \ref{BaEvLe1} to our problem, we start using the following rough estimate
\begin{align}
K_0(\beta |\boldsymbol{x}|)
=\frac {1}{2} \int _{0}^{\infty} e^{-{\frac {\pi \beta^2 |\boldsymbol{x}|^{2}}{\eta}}}
e^{-\frac{\eta}{4\pi}}~\frac{1}{\eta}\,d\eta
&=\frac {1}{2 |\boldsymbol{x}|^\alpha} \int _{0}^{\infty} \frac{|\boldsymbol{x}|^\alpha}{\eta^{\alpha/2}}
e^{-{\frac {\pi \beta^2 |\boldsymbol{x}|^{2}}{\eta}}}
e^{-\frac{\eta}{4\pi}}~\frac{1}{\eta^{1-\alpha/2}}\,d\eta \nonumber \\[3mm]
&\leqslant \frac {1}{2 |\boldsymbol{x}|^\alpha} 
\left(\sup_{t > 0} t^{\alpha/2} e^{-\pi \beta^2 t}\right)
\int _{0}^{\infty} e^{-\frac{\eta}{4\pi}}~\frac{1}{\eta^{1-\alpha/2}}\,d\eta\,\,.
\label{Anous}
\end{align}
Now, since
\begin{enumerate}[(i)]
\item $\displaystyle{t^{\alpha/2} e^{-\pi \beta^2 t} > 0\,\,,\quad \forall\;t >0}$\,\,,
\item $\displaystyle{\lim_{t \to 0^{+}}\left(t^{\alpha/2} e^{-\pi \beta^2 t}\right)=0=\lim_{t \to +\infty}\left(t^{\alpha/2} e^{-\pi \beta^2 t}\right)}$\,\,,
\item $\displaystyle{\frac{d}{dt}\left(t^{\alpha/2} e^{-\pi \beta^2 t}\right)=0 \iff t=\frac{\alpha}{2 \pi \beta^2}}$\,\,,
\end{enumerate}
we conclude that 
\begin{align*}
\sup_{t > 0} t^{\alpha/2} e^{-\pi \beta^2 t}
=t^{\alpha/2} e^{-\pi \beta^2 t}\Bigl|_{t=\frac{\alpha}{2 \pi \beta^2}}
=\left(\frac{\alpha}{2\pi \beta^2}\right)^{\alpha/2} e^{-\alpha/2}\,\,.
\end{align*}
Furthermore, according the Table of Integrals of Gradshtein-Ryzhik~\cite[{\bf 3.381}, $4.$, p.346]{GR}, it follows that
\[
\int _{0}^{\infty} e^{-\frac{\eta}{4\pi}}~\frac{1}{\eta^{1-\alpha/2}}\,d\eta
=(4\pi)^{\alpha/2} \Gamma(\alpha/2)\,\,.
\]
Hence, we have
\begin{align}
K_0(\beta |\boldsymbol{x}|) \leqslant \frac {C_{\alpha,\beta}}{|\boldsymbol{x}|^\alpha}\,\,,
\label{compare}
\end{align}
where
\begin{align*}
C_{\alpha,\beta}=\frac{1}{2} \left(\frac{\alpha}{2\pi \beta^2}\right)^{\alpha/2} 
(4\pi)^{\alpha/2} \Gamma(\alpha/2) e^{-\alpha/2}\,\,.
\end{align*}

If we take $\alpha=1$, then, for $d=2$, it follows from Eq.(\ref{LS1}) that in the international system of units
we have 
\begin{align}
\langle \psi,K_0(\beta |\boldsymbol{x}|) \psi \rangle
&=\int_{\oR^2} \overline{\psi}(\boldsymbol{x}) K_0(\beta \boldsymbol{x}) \psi(\boldsymbol{x})\,d\boldsymbol{x} \nonumber \\[3mm]
&\leqslant \frac{1}{\sqrt{2}} \Gamma(1/2) e^{-1/2} \frac{1}{\beta}
\int_{\oR^2} \frac{|\psi(\boldsymbol{x})|^2}{|\boldsymbol{x}|}\,d\boldsymbol{x} \nonumber \\[3mm]
&< \frac{1}{\sqrt{2}} \Gamma(1/2) e^{-1/2} \left(\frac{\Gamma(1/4)}{\sqrt{2}\,\Gamma(3/4)}\right)^2 \frac{1}{\hbar c \beta}
\left\langle \psi,\sqrt{-\hbar^2 c^2 \Delta}\;\psi \right\rangle \nonumber \\[3mm]
&\leqslant \frac{1}{\sqrt{2}} \Gamma(1/2) e^{-1/2} \left(\frac{\Gamma(1/4)}{\sqrt{2}\,\Gamma(3/4)}\right)^2 \frac{1}{\hbar c \beta}
\left\langle \psi,\sqrt{-\hbar^2 c^2 \Delta+m^2 c^4}\;\psi \right\rangle\,\,.
\label{Kolimited}
\end{align}

Note that we can simplify the expression of the constant in front of $\left\langle \psi,\sqrt{-\hbar^2 c^2 \Delta+m^2 c^4}\;\psi \right\rangle$,
taking into account the relationship that exists between the gamma function and the beta function. Indeed, it follows that
\begin{align*}
\frac{1}{\sqrt{2}} \Gamma(1/2) e^{-1/2} \left(\frac{\Gamma(1/4)}{\sqrt{2}\,\Gamma(3/4)}\right)^2 \frac{1}{\hbar c \beta}
=\frac{[B(1/2,1/4)]^2}{\sqrt{8 \pi e}}\;\frac{1}{\hbar c \beta}
=\frac{[\Gamma(1/4)]^4}{2 (2\pi)^{3/2} e^{1/2}}\;\frac{1}{\hbar c \beta}\,\,.
\end{align*}
In the last equality we use the well-known expression
\begin{align*}
B(x,y)=2 \int_0^{\pi/2} (\cos \varphi)^{2x-1} (\sin \varphi)^{2y-1}\;d\varphi\,\,, 
\end{align*}
and the Table of Integrals of Gradshtein-Ryzhik~\cite[{\bf 3.621}, $7.^*$, p.395]{GR} in order to calculate the value of the 
function $B(1/2,1/4)$.

So, again, the bounded from below problem of the operator $\boldsymbol{\cal B}_{\rm FW}$ is simply to show that
\begin{align*}
\inf\,\langle \psi,\boldsymbol{\cal B}_{\rm FW} \psi \rangle \geqslant 0\,\,, 
\end{align*}
where the infimum is taken over all $\psi \in H^{1/2}(\oR^2;\oC)$. In other words, the operator $\boldsymbol{\cal B}_{\rm FW}$ has to
be positive and, clearly, the necessary condition for stability is captured in the estimate
\begin{align*}
\left\langle \psi,\sqrt{-\hbar^2 c^2\,\Delta+m^2 c^4}\;\psi \right\rangle
-\gamma \langle \psi,K_0(\beta |\boldsymbol{x}|) \psi \rangle \geqslant 0\,\,.
\end{align*}
Consequently, the inequality (\ref{Kolimited}) shows that operator $\boldsymbol{\cal B}_{\rm FW}$ is bounded from below if
\begin{align*}
\left(1-\gamma\;\frac{[\Gamma(1/4)]^4}{2 (2\pi)^{3/2} e^{1/2}}\;\frac{1}{\hbar c \beta}\right)
\left\langle \psi,\sqrt{-\hbar^2 c^2\,\Delta+m^2 c^4}\;\psi \right\rangle \geqslant 0\,\,.
\end{align*}
In other words, $\langle \psi,\boldsymbol{\cal B}_{\rm FW} \psi \rangle$ is lower bounded if
\begin{align*}
\gamma \leqslant \gamma_{\rm c}^{\cal H}
=\left(\frac{2 (2\pi)^{3/2} e^{1/2}}{[\Gamma(1/4)]^4}\right)\;\hbar c \beta\,\,.
\end{align*}
This ends the second proof of the boundedness from below of the Brown-Ravenhall operator with $K_0$-potential.

\subsection{The self-adjoint realization}
\label{Subsec3c}
\hspace*{\parindent}
The self-adjointness of the operator $\boldsymbol{\cal B}_{\rm FW}$ depends on the basic theorem
(see Simon~\cite[Theorem II.7]{Simon} and \cite[Theorem 2]{Simon1}, and also Reed-Simon~\cite[Theorem X.17]{RS}):

\begin{theorem}[KLMN Theorem]
Let $\boldsymbol{\cal H}_0$ be a positive self-adjoint operator and suppose $\langle \psi,V \psi \rangle$
is a symmetric quadratic form such that $|V|^{1/2} \leqslant \boldsymbol{\cal H}_0^{1/2}$, with
$\boldsymbol{\mathfrak{Q}}(V) \supset \boldsymbol{\mathfrak{Q}}(\boldsymbol{\cal H}_0)$, so that for some
$a < 1$ and some $b \in \oR$,
\begin{align*}
|\langle \psi,V \psi \rangle|
\leqslant a \langle \psi,\boldsymbol{\cal H}_0 \psi \rangle+b\;\|\psi\|_2^2\,\,.
\end{align*}
for all $\psi \in \boldsymbol{\mathfrak{Q}}(\boldsymbol{\cal H}_0)$. Then the quadratic form
$\psi \mapsto \langle \psi,\boldsymbol{\cal H}_0 \psi \rangle+\langle \psi,V \psi \rangle$ defined on
$\boldsymbol{\mathfrak{Q}}(\boldsymbol{\cal H}_0) \cap \boldsymbol{\mathfrak{Q}}(V) \equiv
\boldsymbol{\mathfrak{Q}}(\boldsymbol{\cal H}_0)$ is the form of a self-adjoint operator which is bounded
below by $-b$.
\end{theorem}

On applying KLMN Theorems we derive

\begin{proposition}
If $\gamma < \gamma_{\rm c}$, for $\gamma_{\rm c} \in \{\gamma_{\rm c}^{\cal S},\gamma_{\rm c}^{\cal H}\}$,
$\boldsymbol{\cal B}_{\rm FW}$ is self-adjoint on the form domain
$\boldsymbol{\mathfrak{Q}}(\boldsymbol{\cal H}_0)=H^{1/2}(\oR^2;\oC)$ and bounded below by 0.
\label{selfadjoin}
\end{proposition}

\begin{proof}
For $V(\boldsymbol{x})=-\gamma K_0(\beta|\boldsymbol{x}|)$ and $\boldsymbol{\cal H}_0=\sqrt{-\hbar^2 c^2\,\Delta+m^2 c^4}$,
the relativistic inequalities of Sobolev and Hardy show us that $|V|^{1/2} < \boldsymbol{\cal H}_0^{1/2}$ and that
\begin{align*}
|\langle \psi,V \psi \rangle| \leqslant \frac{\gamma}{\gamma_{\rm c}} \langle \psi,\boldsymbol{\cal H}_0 \psi \rangle
\quad \text{(for all $\psi \in H^{1/2}(\oR^2;\oC)$)}\,\,.
\end{align*}
Therefore, taking into account that $\frac{\gamma}{\gamma_{\rm c}} \leqslant 1$, for
$\gamma_{\rm c} \in \{\gamma_{\rm c}^{\cal S},\gamma_{\rm c}^{\cal H}\}$, if we take $\gamma < \gamma_{\rm c}$
the hypotheses of the KLMN Theorem are satisfied and one deduces that the quadratic form (\ref{OpBRFWa}) defined on
$\boldsymbol{\mathfrak{Q}}(\boldsymbol{\cal H}_0) \cap \boldsymbol{\mathfrak{Q}}(V) \equiv \boldsymbol{\mathfrak{Q}}(\boldsymbol{\cal H}_0)$
is the form of the self-adjoint operator, $\boldsymbol{\cal B}_{\rm FW}$, and the inequality holds with $b=0$, {\em i.e.},
the quadratic form $\langle \psi,\boldsymbol{\cal B}_{\rm FW} \psi \rangle$ is bounded below by 0.
\end{proof}

\begin{remark}
For the critical value $\gamma=\gamma_{\rm c}$, the Friedrichs Extension Theorem~\cite[Theorems X.23]{RS}
guarantees that the quadratic form (\ref{OpBRFWa}) is a closable quadratic form and its closure is the
quadratic form of a unique self-adjoint operator associated with $\boldsymbol{\cal B}_{\rm FW}$, 
which is bounded from below, and the lower bound of its spectrum is the same lower bound of (\ref{OpBRFWa}).
\label{FETheo}
\end{remark}

\subsection{Is $\langle \psi,\boldsymbol{\cal B}_{\rm FW} \psi \rangle$ unbounded below if 
$\gamma > \gamma_{\rm c}$?}
\label{Subsec3d}
\hspace*{\parindent}
We proved that $V(\boldsymbol{x})=-\gamma K_0(\beta |\boldsymbol{x}|)$ can be controlled by
$\sqrt{-\hbar^2 c^2\,\Delta+m^2 c^4}$, if $\gamma$ is smaller than $\gamma_{\rm c}^{\cal S}$, or $\gamma_{\rm c}^{\cal H}$.
But, bearing in mind hypothetical applications of the model presented here to graphene, or any other two 
dimensional system, it is important to check if $\langle \psi,\boldsymbol{\cal B}_{\rm FW} \psi \rangle$ is unbounded below if
$\gamma > \gamma_{\rm c}$. For this, we will take the form
\begin{align}
\int_{\oR^2} \Bigl\{\overline{\psi}(\boldsymbol{x}) \sqrt{-\hbar^2 c^2\,\Delta+m^2 c^4}\;\psi(\boldsymbol{x})
-\gamma\,\overline{\psi}(\boldsymbol{x}) K_0(\beta|\boldsymbol{x}|) \psi(\boldsymbol{x})\Bigr\}\,d\boldsymbol{x}\,\,,
\label{BRform}
\end{align}
for all $\psi$ in the form domain of $\boldsymbol{\cal H}_0+V$, and, in particular, for all $\psi \in {\mathscr S}(\oR^2;\oC)$.
As usual, over all normalized function $\psi(\boldsymbol{x})$, we shall replace $\psi(\boldsymbol{x})$ by the scaled function
$\psi_\lambda(\boldsymbol{x})=\lambda \psi(\lambda\boldsymbol{x})$, for $\lambda > 0$. A little computation shows that
the normalization is unchanged.

As we explicitly know the function $K_0$ in Fourier space (see Proposition \ref{FTK0}), using Parseval's formula,
we can rewrite (\ref{BRform}) as follows:
\begin{align}
\frac{1}{(2\pi \hbar)^2} \int_{\oR^2} \sqrt{c^2 |\boldsymbol{p}|^2+m^2 c^4}\,|\widehat{\psi}(\boldsymbol{p})|^2\,d\boldsymbol{p}
-\frac{\gamma}{(2\pi \hbar)^4} \int_{\oR^2} \int_{\oR^2} 
\overline{\widehat{\psi}}(\boldsymbol{p}) \frac{2\pi \hbar^2}{|\boldsymbol{p}-\boldsymbol{q}|^2+\hbar^2 \beta^2}\,
\widehat{\psi}(\boldsymbol{q})\,d\boldsymbol{q} d\boldsymbol{p}\,\,.
\label{BRform2}
\end{align}
Naturally, in Fourier space, we must replace $\widehat{\psi}(\boldsymbol{p})$ by the scaled function 
$\widehat{\psi}_\lambda(\boldsymbol{p})=\lambda^{-1} \widehat{\psi}(\lambda^{-1}\boldsymbol{p})$,
for $\lambda > 0$ and a little computation shows that, again, the normalization is unchanged. So with that in mind, it follows that
\begin{align*}
&\frac{1}{(2\pi \hbar)^2} \int_{\oR^2} \sqrt{c^2 |\boldsymbol{p}|^2+m^2 c^4}\,|\widehat{\psi}_\lambda(\boldsymbol{p})|^2\,d\boldsymbol{p}
-\frac{\gamma}{(2\pi \hbar)^4 } \int_{\oR^2} \int_{\oR^2} 
\overline{\widehat{\psi}}_\lambda(\boldsymbol{p}) \frac{2\pi \hbar^2}{|\boldsymbol{p}-\boldsymbol{q}|^2+\hbar^2 \beta^2}\,
\widehat{\psi}_\lambda(\boldsymbol{q})\,d\boldsymbol{q} d\boldsymbol{p} \\[3mm]
&= \frac{\lambda}{(2\pi \hbar)^2} \int_{\oR^2} \sqrt{c^2 |\boldsymbol{p}|^2+\frac{m^2 c^4}{\lambda^2}}\,|\widehat{\psi}(\boldsymbol{p})|^2\,d\boldsymbol{p}
-\frac{\gamma}{(2\pi \hbar)^4 } \int_{\oR^2} \int_{\oR^2} 
\overline{\widehat{\psi}}(\boldsymbol{p}) \frac{2\pi \hbar^2}{|\boldsymbol{p}-\boldsymbol{q}|^2+\frac{\hbar^2 \beta^2}{\lambda^2}}\,
\widehat{\psi}(\boldsymbol{q})\,d\boldsymbol{q} d\boldsymbol{p}\,\,.
\end{align*}
Now, on allowing $\lambda \to \infty$, we obtain that the kinetic term {\em will always dominate}
$V(\boldsymbol{x})=-\gamma K_0(\beta |\boldsymbol{x}|)$, regardless of the value of the coupling constant
$\gamma$. In other words, the quadratic form
$\langle \psi,\boldsymbol{\cal B}_{\rm FW} \psi \rangle$ does not cease to be bounded below even if $\gamma > \gamma_{\rm c}$.
This result is a characteristic of the $K_0$-potential. The reason for this behavior is that, unlike to usual Brown-Ravenhall operator
with coulombian potential, the Brown-Ravenhall operator with $K_0$-potential it is not homogeneous with respect to scalings of
$\oR^2$, {\em i.e.}, the kinetic energy does not have the same behavior under scaling as the Bessel-Macdonald energy
for large momenta. Consequently, a {\em complete implosion will never occur} for two relativistic fermions interacting via an
attractive potential of the Bessel-Macdonald type.

\begin{remark}[{\bf The role of mass in the scaling procedure}]
Note that we lost the masses in the scaling procedure. It is irrelevant for the stability problem, but it will determine 
the energy once the problem is shown to be stable. Indeed, the simple inequality $\sqrt{c^2 |\boldsymbol{p}|^2+m^2 c^4}
\geqslant c |\boldsymbol{p}|$ shows that stability also holds in the case of non-zero mass whenever it holds with
zero mass, as we have seen through relativistic inequalities of Sobolev and Hardy.
\end{remark}

\begin{remark}[{\bf The atomic collapse in graphene}]
Non-perfect graphene has additional potentials; a particular case of importance is the presence of an impurity of
Coulomb type. In this case, there are two distinct regimes: the {\em subcritical regime} and the {\em supercritical regime}.
In the latter case, when $\gamma > \gamma_{\rm c}$, a phenomenon called {\em atomic collapse in graphene} occurs.
This result physically means that the electron is pulled into the nucleus of the interstitial atom. The supercritical instability
in the field of a single charged impurity has been extensively studied in the literature (see~\cite{Gusynin1} and references therein).
On the other hand, as we have seen, since the kinetic energy for large momenta (or for small distances) is always larger
than the Bessel-Macdonald type interaction energy, this implies that the phenomenon of the fall-to-center should not take place
in graphene in the presence of an impurity of Bessel-Macdonald type (a complete implosion would imply the annihilation of matter).
Despite this, by Proposition \ref{selfadjoin} and Remark \ref{FETheo}, only for values of $\gamma \leqslant \gamma_{\rm c}$
the operator $\boldsymbol{\cal B}_{\rm FW}$ is self-adjoint. Thereby, the critical coupling constant as occurring in this paper can
be mathematically thought of as that coupling constant where a natural definition of self-adjointness ceases to exist. In other words,
by analyzing the underlying trapping mechanism for the model proposed here, {\em a priori}, the supercritical instability in graphene
with an impurity of Bessel-Macdonald type should not be necessarily related to the phenomenon of the fall-to-center, but must be
interpreted as the absence of dynamics.
\end{remark}

\section{Spectral properties of the operator $\boldsymbol{\cal B}$}
\label{Sec4}
\hspace*{\parindent}
In this section, we investigate the nature of the spectrum of the Brown-Ravenhall operator with $K_0$-potential.
We start by locating the essential spectrum, $\sigma_{\rm ess}(\boldsymbol{\cal B})$. Firstly, we note that the map
$\langle \Psi,\boldsymbol{\cal B} \Psi \rangle \to \langle \psi,\boldsymbol{\cal B}_{\rm FW} \psi \rangle$, where
$\Psi \in L_2(\oR^2;\oC^2)$ and $\psi \in L_2(\oR^2;\oC)$, determines a unitary equivalence between the operators
$\boldsymbol{\cal B}$ and $\boldsymbol{\cal B}_{\rm FW}$, hence, they have the same spectral properties.
This leads us to the following result:

\begin{theorem}
Assume that $0 < \gamma \leqslant \gamma_{\rm c}$. Then for the essential spectrum of the Brown-Ravenhall
operator (\ref{OpBR}) one has $\sigma_{\rm ess}(\boldsymbol{\cal B})=\sigma_{\rm ess}(\boldsymbol{\cal B}_{\rm FW})
=[mc^2,\infty)$.
\label{SpecTheo1}
\end{theorem}

The proof of Theorem \ref{SpecTheo1} depends fundamentally on the following

\begin{lemma}
The potential $V(\boldsymbol{x})=-\gamma K_0(\beta |\boldsymbol{x}|)  \in L_2(\oR^2)+L_\infty(\oR^2)_{\varepsilon}$,
where the $\varepsilon$ indicates that for any $\varepsilon > 0$, we can decompose $V=V_1+V_2$ with 
$V_1 \in L_2(\oR^2)$ and $V_2 \in L_\infty(\oR^2)$, with $\|V_2\|_\infty < \varepsilon$.
\label{SpecLemma1}
\end{lemma}

\begin{proof}
For any $\varepsilon > 0$, let $\chi_\varepsilon(|\boldsymbol{x}|)$ be the function that is $1$ on 
$\{\boldsymbol{x} \mid |\boldsymbol{x}| \leqslant (\gamma \varepsilon)^{-1}\}$ and that vanishes
outside $\{\boldsymbol{x} \mid |\boldsymbol{x}| < 2(\gamma \varepsilon)^{-1}\}$. Then, we decompose the
potential $V(\boldsymbol{x})$ as 
\[
V(\boldsymbol{x})=-\gamma \chi_\varepsilon(|\boldsymbol{x}|) K_0(\beta |\boldsymbol{x}|)
-\gamma \bigl(1-\chi_\varepsilon(|\boldsymbol{x}|)\bigr) K_0(\beta |\boldsymbol{x}|)
=V_1(\boldsymbol{x})+V_2(\boldsymbol{x})\,\,.
\]
By using the Eq.(\ref{Grads}), we obtain
\begin{align*}
\|V_1\|_2=\left(\gamma^2 \int_{\oR^2}
|\chi_\varepsilon(|\boldsymbol{x}|) K_0(\beta |\boldsymbol{x}|)|^2\,d^2\boldsymbol{x}\right)^{1/2}
=\left(2 \pi \gamma^2 \int_{0}^\infty \chi^2_\varepsilon(r) K^2_{0}(\beta r)\,r\,dr\right)^{1/2}
< \frac{\gamma \pi^{1/2}}{\beta}\,\,.
\end{align*}
Hence, $V_1 \in L_2(\oR^2)$. In the sequel we will use the asymptotic behavior of $K_0$, {\em i.e.},
\begin{align*}
K_0(\beta |\boldsymbol{x}|) \simeq \sqrt{\frac{\pi}{2 \beta |\boldsymbol{x}|}}~e^{-\beta |\boldsymbol{x}|}\,\,,
\quad |\boldsymbol{x}| \to \infty\,\,.
\end{align*}
This implies that $V_2$ is a bounded function vanishing at $\infty$. So it follows that
$V_2 \in L_\infty(\oR^2)$, with
\begin{align*}
\|V_2\|_\infty=\sup_{\boldsymbol{x} \in \oR^2} 
\bigl|\gamma \bigl(1-\chi_\varepsilon(|\boldsymbol{x}|)\bigr) K_0(\beta |\boldsymbol{x}|)\bigr|
< \varepsilon\,\,.
\end{align*}
Therefore, the potential $V(\boldsymbol{x})=-\gamma K_0(\beta |\boldsymbol{x}|) \in L_2(\oR^2)+L_\infty(\oR^2)_{\varepsilon}$.
\end{proof}

\begin{proof}[Proof of Theorem \ref{SpecTheo1}]
Let us start by defining $\boldsymbol{\cal B}_{\rm FW0} \equiv \boldsymbol{\cal H}_0=\sqrt{-\hbar^2 c^2\,\Delta+m^2 c^4}$.
From Remark \ref{Note01}, we know that $\sigma_{\rm ess}(\boldsymbol{\cal B}_{\rm FW0})=[mc^2,\infty)$. On the other hand,
in order to locate $\sigma_{\rm ess}(\boldsymbol{\cal B}_{\rm FW})$, where $\boldsymbol{\cal B}_{\rm FW}=\boldsymbol{\cal B}_{\rm FW0}+V$
(with $V(\boldsymbol{x})=-\gamma K_0(\beta |\boldsymbol{x}|)$), we study the resolvent operator
$(\boldsymbol{\cal B}_{\rm FW}-\lambda {\un})^{-1}$ for a some $\lambda \notin \sigma(\boldsymbol{\cal B}_{\rm FW})$.

By the second resolvent equation, for any value of $\lambda \in \rho(\boldsymbol{\cal B}_{\rm FW}) \cap \rho(\boldsymbol{\cal B}_{\rm FW0})$, where
$\rho(\boldsymbol{\cal B}_{\rm FW})$ and $\rho(\boldsymbol{\cal B}_{\rm FW0})$ are the resolvent sets of $\boldsymbol{\cal B}_{\rm FW}$
and $\boldsymbol{\cal B}_{\rm FW0}$, respectively,
we have
\begin{align}
(\boldsymbol{\cal B}_{\rm FW}-\lambda{\un})^{-1}-(\boldsymbol{\cal B}_{\rm FW0}-\lambda{\un})^{-1}
=(\boldsymbol{\cal B}_{\rm FW}-\lambda{\un})^{-1}V(\boldsymbol{\cal B}_{\rm FW0}-\lambda{\un})^{-1}\,\,,
\label{SIRHamil}
\end{align}
(we recall that since $\boldsymbol{\cal B}_{\rm FW}$ is a positive self-adjoint operator, once $\langle \psi,\boldsymbol{\cal B}_{\rm FW} \psi \rangle$
is bounded from below if $\gamma \leqslant \gamma_{\rm c}$, $\lambda \in \rho(\boldsymbol{\cal B}_{\rm FW})$ if and only if
$(\boldsymbol{\cal B}_{\rm FW}-\lambda{\un}):\boldsymbol{\mathfrak{Dom}}(\boldsymbol{\cal B}_{\rm FW}) \to {\mathscr H}_+$ is bijective and its inverse
is bounded). We will show that $(\boldsymbol{\cal B}_{\rm FW}-\lambda{\un})^{-1}-(\boldsymbol{\cal B}_{\rm FW0}-\lambda{\un})^{-1}$ is compact as an
operator on $L_2(\oR^2;\oC)$ and therefore $\sigma_{\rm ess}(\boldsymbol{\cal B}_{\rm FW})=\sigma_{\rm ess}(\boldsymbol{\cal B}_{\rm FW0})=[mc^2,\infty)$
by Weyl's criterion~\cite[Theorem XIII.14]{RS4}.

In view of the self-adjointness of $\boldsymbol{\cal B}_{\rm FW}$, the half-planes $\oC_\pm$ lie in $\rho(\boldsymbol{\cal B}_{\rm FW})$ and
$(\boldsymbol{\cal B}_{\rm FW}-\lambda{\un})^{-1}$ is bounded by
\begin{align*}
\|(\boldsymbol{\cal B}_{\rm FW}-\lambda{\un})^{-1}\| \leqslant |{\rm Im}\,\lambda|^{-1}\,\,,
\quad \lambda \notin \oR\,\,, 
\end{align*}
({\em cf.} Ref.\cite[Corollary 5.7]{HS}). Thus, it remains for us to show that $V(\boldsymbol{\cal B}_{\rm FW0}-\lambda{\un})^{-1}$ is compact. 

Taking into account the basic fact of inclusion~\cite[Chapter 7]{LM}, 
\begin{align*}
H^1(\oR^2;\oC) \subset H^{1/2}(\oR^2;\oC)\,\,, 
\end{align*}
{\em i.e.}, since the range of $(\boldsymbol{\cal B}_{\rm FW0}-\lambda{\un})^{-1}$, namely
$\boldsymbol{\mathfrak{Dom}}(\boldsymbol{\cal H}_0)=H^1(\oR^2;\oC)$,
lies in the form domain $\boldsymbol{\mathfrak{Q}}(\boldsymbol{\cal H}_0)=H^{1/2}(\oR^2;\oC)$, we just show that
$V(\boldsymbol{x})=-\gamma K_0(\beta |\boldsymbol{x}|)$ is a compact operator from $H^1(\oR^2;\oC)$ to $L_2(\oR^2;\oC)$.
This is enough to guarantee that the perturbation $V$ does not modify the essential spectrum of the operator $H$. For this
we will need the following compactness theorem, the proof of which can be found in~\cite[Theorem 1.10]{HiLa}.

\begin{theorem}[Rellich compactness criterion]
Let $H^1_0(\Omega)$ be the closure of $C^\infty_0(\Omega)$ in $H^1(\Omega)$, where $\Omega \subset \oR^d$ 
is an open and bounded set. Then for any bounded sequence $(\psi_\ell)_{\ell \in \oN}$ in $H^1_0(\Omega)$
there exists a subsequence $(\psi_{\ell_k})_{k \in \oN} \subset (\psi_\ell)_{\ell \in \oN}$ such that 
$(\psi_{\ell_k})_{k \in \oN}$ converges strongly in $L_2(\Omega)$.
\label{Rellich}
\end{theorem}

\begin{remark}
$H^1_0(\Omega)$, with the scalar product $\langle\,\cdot\,,\,\cdot\,\rangle_{H^1(\Omega)}$ defined by
\begin{align*}
\langle \psi,\varphi \rangle_{H^1(\Omega)}=\langle \psi,\varphi \rangle_{L_2(\Omega)}
+\sum_{\kappa=1}^d\langle D_\kappa \psi,D_\kappa \varphi \rangle_{L_2(\Omega)}\,\,,
\end{align*} is a Hilbert space by construction, being a closed subspace of the Hilbert space
$H^1(\Omega)$ ({\em cf.}~\cite[Lemma 22.1]{Eskin}).
\end{remark}

Returning to the proof of Theorem \ref{SpecTheo1}, select a sequence $(\varphi_\ell)_{\ell \in \oN} \subset L_2(\oR^2;\oC)$
such that $\varphi_\ell \overset{\boldsymbol{\rm w}}{\longrightarrow} 0$. Assume that
\begin{align*}
\psi_\ell=(\boldsymbol{\cal B}_{\rm FW0}-\lambda{\un})^{-1}\varphi_\ell\,\,, 
\end{align*}
where $\psi_\ell \in \boldsymbol{\mathfrak{Dom}}(\boldsymbol{\cal B}_{\rm FW0})$. The fact that $(\varphi_\ell)_{\ell \in \oN}$ converges weakly
implies that $(\varphi_\ell)_{\ell \in \oN}$ is bounded in $L_2(\oR^2;\oC)$. As $(\boldsymbol{\cal B}_{\rm FW0}-\lambda{\un})^{-1}$ is bounded,
the sequence $(\psi_\ell)_{\ell \in \oN}$ is also bounded in $\boldsymbol{\mathfrak{Dom}}(\boldsymbol{\cal B}_{\rm FW0})$ and converges weakly
to zero in $\boldsymbol{\mathfrak{Dom}}(\boldsymbol{\cal B}_{\rm FW0}) \subset L_2(\oR^2;\oC)$.

Now, according to Lemma \ref{SpecLemma1}, for any $\varepsilon > 0$, we can decompose $V$ as $V=V_1+V_2$,
where $V_1 \in L_2(\oR^2)$ and $V_2 \in L_\infty(\oR^2)_\varepsilon$. Then, applying the H\"older Inequality, 
\begin{align}
\|V_2 \psi_\ell\|_2 \leqslant \|V_2\|_\infty \|(\boldsymbol{\cal B}_{\rm FW0}-\lambda{\un})^{-1}\varphi_\ell\|_2
\leqslant \varepsilon \|\varphi_\ell\|_2\,\,,
\label{4.13}
\end{align}
and it follows that $\|V_2(\boldsymbol{\cal B}_{\rm FW0}-\lambda{\un})^{-1}\| \leqslant \varepsilon$. Next, observe that $V_1$
is bounded so that there is $K \geqslant 0$ such that $|V_1(\boldsymbol{x})| \leqslant K$, for all $\boldsymbol{x} \in \oR^2$.
Let $B_R=\{\boldsymbol{x} \in \oR^2 \mid |\boldsymbol{x}| < R\}$ be the open ball with radius $R > 0$ around zero in
$\oR^2$ and $\theta_R \in C^\infty_0(B_{2R})$ be such that $0 \leqslant \theta_R \leqslant 1$,
$\theta_R \upharpoonright_{B_R}=1$. Then, $(\theta_R \psi_\ell)_{\ell \in \oN}$ is bounded in $H^1(B_{2R};\oC)$.
We recall that Rellich's Compactness Theorem gives us that the inclusion
\begin{align*}
H^1(B_{2R};\oC) \hookrightarrow L_2(B_{2R};\oC)\,\,,
\end{align*}
is compact. This means that there exists a subsequence $(\theta_R \psi_{\ell_k})_{k \in \oN} \subset (\theta_R \psi_\ell)_{\ell \in \oN}$
such that $(\theta_R \psi_{\ell_k})_{k \in \oN}$ converges strongly in $L_2(B_{2R};\oC)$. We shall show that
$\|V_1 \psi_\ell\|$ converges to 0. With the help of the function $\theta_R$ defined above, it follows that
\begin{align}
\|V_1 \psi_\ell\|_2 \leqslant \|V_1 \theta_R \psi_\ell\|_2+\|V_1 (1-\theta_R) \psi_\ell\|_2
\leqslant \|V_1 \theta_R \psi_\ell\|_2+\|V_1 (1-\theta_R)\| \|\psi_\ell\|_2\,\,.
\label{2.13}
\end{align}
The first term can be made smaller than $\varepsilon$ by choosing $\ell$ large since $V_1$ is bounded, {\em i.e.},
there is $\ell_0 \in \oN$ such that $\|\theta_R \psi_\ell\| \leqslant \varepsilon/K$ for all $\ell \geqslant \ell_0$. As $\psi_\ell$
converges weakly to zero, there is a positive constant $M$ such that $\|\psi_\ell\|_2 \leqslant M$. Hence, by assumption,
the second term can be made smaller than $\varepsilon$ times a positive constant by choosing $R$ large. Consequently,
for $\ell$ large enough, the right-hand side of (\ref{2.13}) is smaller than $\varepsilon$ times a positive constant. This
implies that $V_1 \psi_\ell \to 0$ in $L_2(\oR^2;\oC)$ as $\ell \to \infty$, since $\varepsilon$ is arbitrary, and hence
$V_1(\boldsymbol{\cal B}_{\rm FW0}-\lambda{\un})^{-1}$ is compact.

Finally, by (\ref{4.13}) we have
\begin{align*}
\|V(\boldsymbol{\cal B}_{\rm FW0}-\lambda{\un})^{-1}-V_1(\boldsymbol{\cal B}_{\rm FW0}-\lambda{\un})^{-1}\| < \varepsilon\,\,,
\end{align*}
so $V(\boldsymbol{\cal B}_{\rm FW0}-\lambda{\un})^{-1}$ is approximated by compact operators and is itself compact~\cite[Theorem 9.8]{HS}.
Thus, we have that $(\boldsymbol{\cal B}_{\rm FW}-\lambda{\un})^{-1}-(\boldsymbol{\cal B}_{\rm FW0}-\lambda{\un})^{-1}$
is equal to the product of a compact operator and a bounded operator, and since the product of a compact operator with a bounded
operator is compact~\cite[Theorem 9.5(b)]{HS}, this proves the required compactness on $L_2(\oR^2;\oC)$. This means that
$\boldsymbol{\cal B}_{\rm FW}$ and $\boldsymbol{\cal B}_{\rm FW0}$ have the same essential spectrum by Weyl's criterion, namely
$\sigma_{\rm ess}(\boldsymbol{\cal B}_{\rm FW})=\sigma_{\rm ess}(\boldsymbol{\cal B}_{\rm FW0})=[mc^2,\infty)$~\cite[Theorem XIII.14]{RS4}.
\end{proof}

Possible embedded eigenvalues in the essential spectrum, $\sigma_{\rm ess}(\boldsymbol{\cal B}_{\rm FW})$, are not stable.
Unlike isolated eigenvalues, which are relatively stable under perturbations, embedded eigenvalues generically disappear from
the spectrum of a self-adjoint operator under perturbations (actually, they do not disappear but, they become {\em resonances}
of the operator).\footnote{An excellent mathematical, comprehensive and self-contained analysis of quantum resonances can
be found in~\cite{HS}.} The study of resonances is closely connected with the absence of eigenvalues in the essential spectrum.
In particular, for operator $\boldsymbol{\cal B}_{\rm FW}$, there are no eigenvalues embedded in the interior of the essential spectrum
when $\gamma \leqslant \gamma_{\rm c}$, with $\gamma_{\rm c} \in \{\gamma_{\rm c}^{\cal S},\gamma_{\rm c}^{\cal H}\}$.
This result is proved with the help of the following

\begin{lemma}[An abstract virial theorem~\cite{BaEv}, Lemma 3.1.2]
Let ${\mathscr U}(a)$, $a \in \oR_+$ , be a one parameter family of unitary operators on a Hilbert space
${\mathscr H}$ which converges strongly to the identity as $a \to 1$. Let $T$ be a self-adjoint operator in 
${\mathscr H}$ and $T_a=f(a){\mathscr U}(a) T {\mathscr U}^{-1}(a)$ where $f(1)=1$ and
$f'(1)$ exists. If $\psi \in \boldsymbol{\mathfrak{Dom}}(T) \cap \boldsymbol{\mathfrak{Dom}}(T_a)$
is an eigenvector of $T$ corresponding to an eigenvalue $\lambda$ then
\begin{align*}
\lim_{a \to 1} \left\langle \psi_a,\left[\frac{T_a-T}{a-1}\right]\psi \right\rangle
=\lambda f'(1) \|\psi\|_2^2\,\,,
\end{align*}
where $\psi_a={\mathscr U}(a)\psi$.
\label{BaEvans}
\end{lemma}

The use of Lemma \ref{BaEvans} with $T=\boldsymbol{\cal B}_{\rm FW}$, $f(a)= a$, ${\mathscr U}(a)$ the unitary operator
defined by ${\mathscr U}(a)\psi(\boldsymbol{x})=a \psi(a\boldsymbol{x})$, combined with Theorem 3.2.3 in Ref~\cite{BaEv},
which can be adapted to the case of the operator $\boldsymbol{\cal B}_{\rm FW}$, guarantees the absence of discrete spectrum
in $[mc^2,\infty)$.

On the other hand, in many applications, a self-adjoint operator has a number of eigenvalues below the bottom of the
essential spectrum. In the sequel, we will show that the stable states of operator $\boldsymbol{\cal B}$ are represented by
eigenvalues lying in the gap of the essential spectrum of $\boldsymbol{\cal B}_{\rm FW}$, {\em i.e.},
in the interval $[0,mc^2)$. We recall that the eigenvalues of $\boldsymbol{\cal B}_{\rm FW}$
are characterized by the {\em min-max principle}. This theorem establishes that since $\boldsymbol{\cal B}_{\rm FW}$
is self-adjoint, and if $\lambda_1 \leqslant \lambda_2 \leqslant \lambda_3 \cdots$ are eigenvalues of
$\boldsymbol{\cal B}_{\rm FW}$ below the essential spectrum, respectively, the infimum of the essential spectrum,
once there are no more eigenvalues left, then
\begin{align*}
\lambda_n=\inf_{\psi,\ldots,\psi_{n-1}} \sup_{\psi \in \Omega(\psi_1,\ldots,\psi_n)} \langle \psi,\boldsymbol{\cal B}_{\rm FW} \psi \rangle\,\,,
\end{align*}
where $\Omega(\psi_1,\ldots,\psi_n)=\bigl\{\psi \in \boldsymbol{\mathfrak{Dom}}(\boldsymbol{\cal B}_{\rm FW}) \mid \|\psi\|_2=1,
\psi \in {\rm span}\{\psi_1,\ldots,\psi_n\}\bigr\}$. Hence, if there exists $\psi \in \boldsymbol{\mathfrak{Dom}}(\boldsymbol{\cal B}_{\rm FW})$
such that $\langle \psi,\boldsymbol{\cal B}_{\rm FW} \psi \rangle < mc^2$, then $\boldsymbol{\cal B}_{\rm FW}$ has at least one
eigenvalue below the bottom of the essential spectrum, $\sigma_{\rm ess}(\boldsymbol{\cal B}_{\rm FW})$. Indeed, if this were
not true then $\sigma(\boldsymbol{\cal B}_{\rm FW}) \cap (0,mc^2)=\varnothing$ meaning that
$\sigma(\boldsymbol{\cal B}_{\rm FW}) \subset [mc^2,\infty)$. By the spectral theorem, this would imply that
$\boldsymbol{\cal B}_{\rm FW} \geqslant mc^2$, {\em i.e.}, $\langle \psi,\boldsymbol{\cal B}_{\rm FW} \psi \rangle \geqslant mc^2$
for all $\psi \in \boldsymbol{\mathfrak{Dom}}(\boldsymbol{\cal B}_{\rm FW})$ in contradiction to the assumption
$\langle \psi,\boldsymbol{\cal B}_{\rm FW} \psi \rangle < mc^2$. Moreover,
\begin{align*}
\langle \psi,\boldsymbol{\cal B}_{\rm FW} \psi \rangle
&=\langle \psi,\sqrt{-\hbar^2 c^2\,\Delta+m^2 c^4}\;\psi \rangle
-\gamma \langle \psi,K_0(\beta|\boldsymbol{x}|) \psi \rangle \\[3mm]
&\geqslant mc^2 \|\psi\|_2^2-\frac{\gamma}{\gamma_{\rm c}}mc^2 \|\psi\|_2^2 \\[3mm]
&=mc^2 \left(1-\frac{\gamma}{\gamma_{\rm c}}\right) \|\psi\|_2^2\,\,.
\end{align*}
Therefore, for $\psi \in H^{1/2}(\oR^2;\oC)$ (with $\|\psi\|_2=1$) and
$\gamma \leqslant \gamma_{\rm c}$, we obtain
\begin{align*}
\langle \psi,\boldsymbol{\cal B}_{\rm FW} \psi \rangle
\geqslant mc^2 \left(1-\frac{\gamma}{\gamma_{\rm c}}\right)
\quad \text{(for $\gamma_{\rm c} \in \{\gamma_{\rm c}^{\cal S},\gamma_{\rm c}^{\cal H}\}$)}\,\,.
\end{align*}
Thus in $[0,mc^2)$, the spectrum of $\boldsymbol{\cal B}_{\rm FW}$ consists only of eigenvalues. The question
arises whether $\boldsymbol{\cal B}_{\rm FW}$ has a finite or infinite number of eigenvalues? To answer this question
it will be sufficient to take $\hbar=m=c=1$, when $m > 0$, because of scaling. As can be seen from (\ref{compare}),
the operator $\sqrt{-\Delta+1}-\gamma C_{1,\beta}|\boldsymbol{x}|^{-1}$ is an upper bound on the operator
$\boldsymbol{\cal B}_{\rm FW}$ (see (\ref{OpBRFWa})). In turn, taking into account that  $\sqrt{-\Delta+1} \leqslant -\Delta+1$
and that $-\Delta-\gamma C_{1,\beta}|\boldsymbol{x}|^{-1}$ has an infinite number of eigenvalues (see~\cite[Theorem XI.1.5]{EdEvans}),
this ensures the existence of infinitely many eigenvalues smaller than 1 (or rather, smaller than $mc^2$) for
$\boldsymbol{\cal B}_{\rm FW}$ for $\gamma \leqslant \gamma_{\rm c}$, with
$\gamma_{\rm c} \in \{\gamma_{\rm c}^{\cal S},\gamma_{\rm c}^{\cal H}\}$. In conclusion, all of these observations
prove the following

\begin{proposition}
Consider the operator $\boldsymbol{\cal B}_{\rm FW}$ defined by the form in (\ref{OpBRFWa}). If
$\gamma \leqslant \gamma_{\rm c}$, with $\gamma_{\rm c} \in \{\gamma_{\rm c}^{\cal S},\gamma_{\rm c}^{\cal H}\}$,
then
\begin{align*}
\langle \psi,\boldsymbol{\cal B}_{\rm FW} \psi \rangle
\geqslant mc^2 \left(1-\frac{\gamma}{\gamma_{\rm c}}\right)\,\,,
\quad \|\psi\|_2=1\,\,.
\end{align*}
In $[0,mc^2)$, the discrete spectrum $\{\lambda_n\}_{n \geqslant 1} \in \sigma_{\rm disc}(\boldsymbol{\cal B}_{\rm FW})
=\sigma_{\rm disc}(\boldsymbol{\cal B})$, consists of an infinite number of isolated eigenvalues of finite multiplicity.
Moreover, since the eigenvalues $\lambda_n$ are all smaller than $mc^2=\inf\,\{\sigma_{\rm ess}(\boldsymbol{\cal B}_{\rm FW})\}$
and since eigenvalues can accumulate only at the essential spectrum, we have that
\begin{align*}
\lambda_1 \leqslant \cdots \leqslant \lambda_{n-1} \leqslant \lambda_n \to \inf\,\{\sigma_{\rm ess}(\boldsymbol{\cal B}_{\rm FW})\}
=mc^2 \quad \text{for} \quad n \to +\infty\,\,.
\end{align*}
\label{discAA}
\end{proposition}

\begin{remark}
Motivated by the conduction properties of graphene discovered and studied in the last decades, in Proposition \ref{discAA},
the velocity of light $c$ could be replaced by the Fermi velocity in graphene $v_{\rm F} \approx 10^6\,m/s=c/300$.
In this case, the interval $[0, mv_{\rm F}^2)$ is associated a mass-gap effect~\cite{Gusynin,Fuhrer}, observed in pure
monolayer graphene on substrates. At this point, we also want to mention a physical situation where resonances have been
shown to exist. Considering the problem of impurity states in the vicinity of the mass-gap on graphene, taking into account
the {\em spin-orbit (SO) interaction}, Inglot-Dugaev~\cite{Inglot} have shown that even though the internal SO interaction
is relatively small, its effect is crucial because a very small perturbation potential can create both discrete and resonance
impurity states located near the gap. 
\end{remark}

\,\,\,We conclude with

\begin{proposition}
The operator $\boldsymbol{\cal B}_{\rm FW}$ defined by the form in (\ref{OpBRFWa}) has no eigenvalue at 0 if
$\gamma=\gamma_{\rm c}^{\cal H}$.
\end{proposition}

\begin{proof}
Suppose that 0 is an eigenvalue of $\boldsymbol{\cal B}_{\rm FW}$ with corresponding eigenfunction $\psi$. Since the right-hand side
of 
\begin{align*}
\langle \psi,\boldsymbol{\cal B}_{\rm FW} \psi \rangle
=\langle \psi,\sqrt{-\hbar^2 c^2\,\Delta+m^2 c^4}\;\psi \rangle
-\gamma_{\rm c}^{\cal H} \langle \psi,K_0(\beta|\boldsymbol{x}|) \psi \rangle\,\,,
\end{align*}
is non-negative, it must be zero. But this would imply that there is equality in (\ref{LS1}) with $\psi \not= 0$, 
which is not possible.
\end{proof}

\section*{Author's Contributions}
\hspace*{\parindent}
All authors contributed equally to this work. On behalf of all authors, the corresponding author states that
there is no conflict of interest. 

\section*{Data Availability}
\hspace*{\parindent}
The data that support the findings of this study are available from the corresponding author upon reasonable request.

\renewcommand{\theequation}{A.\arabic{equation}}
\renewcommand{\thesection}{A}
\setcounter{equation}{0}

\section*{Appendix A: Boundedness from below via an HSM-type  inequality}
\hspace*{\parindent}
We address in this appendix the existence of an inequalitiy involving the Hardy and Sobolev relativistic inequalities
applied to the Brown-Ravenhall operator $\boldsymbol{\cal B}_{\rm FW}$ with an attractive potential of the Bessel-Macdonald
type. This inequality, that we will call Hardy-Sobolev-Maz'ya type inequality (or HSM-type inequality for brevity)~\cite{BaEvLe},
combines the relativistic Hardy inequality and relativistic Sobolev inequality. Indeed, from the Hardy and Sobolev relativistic
inequalities for $\oR^2$
\begin{align*}
&{\cal H}_2^{-2} \int_{\oR^2} \frac{|\psi(\boldsymbol{x})|^2}{|\boldsymbol{x}|}\;d\boldsymbol{x}
\leqslant \frac{1}{\hbar c} \langle \psi. \sqrt{-\hbar^2 c^2 \Delta}\;{\psi} \rangle
\quad \text{and} \quad
{\cal S}_{2,1/2}^{-1} \|\psi\|_4^2 \leqslant 
\frac{1}{\hbar c} \langle \psi, \sqrt{-\hbar^2 c^2 \Delta}\;{\psi} \rangle\,\,,
\end{align*}
it follows that for $0 < \alpha \leqslant {\cal H}_2^{-2}$
\begin{align}
\frac{1}{\hbar c} \langle \psi, \sqrt{-\hbar^2 c^2 \Delta}\;{\psi} \rangle
-\alpha \int_{\oR^2} \frac{|\psi(\boldsymbol{x})|^2}{|\boldsymbol{x}|}\;d\boldsymbol{x}
&\geqslant
\frac{1}{\hbar c} \left(1-\alpha {\cal H}_2^2\right) \langle \psi, \sqrt{-\hbar^2 c^2 \Delta}\;{\psi} \rangle \nonumber \\[3mm]
&\geqslant \left(1-\alpha {\cal H}_2^2\right) {\cal S}_{2,1/2}^{-1} \|\psi\|_4^2\,\,,
\quad \psi \in H^{1/2}(\oR^2;\oC)\,\,.
\label{HSM}
\end{align}
that is,
\begin{align*}
\frac{1}{\hbar c} \langle \psi, \sqrt{-\hbar^2 c^2 \Delta}\;{\psi} \rangle
-\alpha \int_{\oR^2} \frac{|\psi(\boldsymbol{x})|^2}{|\boldsymbol{x}|}\;d\boldsymbol{x}
-\left(1-\alpha {\cal H}_2^2\right) {\cal S}_{2,1/2}^{-1} \|\psi\|_4^2 \geqslant 0\,\,,
\quad \psi \in H^{1/2}(\oR^2;\oC)\,\,.
\end{align*}
The inequality (\ref{HSM}) says that even when Hardy is subtracted, the operator
$\sqrt{-\hbar^2 c^2 \Delta}$ is still powerful enough to dominate the square of an $L_4$-norm. Therefore,
with this result, for the Brown-Ravenhall operator $\boldsymbol{\cal B}_{\rm FW}$ with an attractive potential of the
Bessel-Macdonald type it follows that
\begin{align}
\langle \psi,K_0(\beta |\boldsymbol{x}|) \psi \rangle
&\leqslant {\cal C}_{\alpha}^{-1} \frac{1}{\hbar c \beta}\;\langle \psi, \sqrt{-\hbar^2 c^2 \Delta}\;{\psi} \rangle \nonumber \\[3mm]
&\leqslant {\cal C}_{\alpha}^{-1} \frac{1}{\hbar c \beta}\;\langle \psi, \sqrt{-\hbar^2 c^2 \Delta+m^2 c^4}\;{\psi} \rangle\,\,,
\label{HSM1}
\end{align}
where
\begin{align*}
{\cal C}_{\alpha}=\alpha \frac{\sqrt{2 e}}{[\Gamma(1/2)]^2}
+\frac{1}{\sqrt{\pi}} \left(1-\alpha {\cal H}_2^2\right) {\cal S}_{2,1/2}^{-1}\,\,.
\end{align*}
Although the constant ${\cal C}_{\alpha}$ is not sharp, we hope to find the best constant ${\cal C}_{\alpha}$
optimizing over the choice of $\alpha$. Consequently, the inequality (\ref{HSM1})
shows that operator $\boldsymbol{\cal B}_{\rm FW}$ is bounded from below if
\begin{align*}
\left(1-\gamma\;{\cal C}_{\alpha}^{-1} \frac{1}{\hbar c \beta}\right) 
\left\langle \psi,\sqrt{-\hbar^2 c^2\,\Delta+m^2 c^4}\;\psi \right\rangle \geqslant 0\,\,,
\end{align*}
In other words, $\langle \psi,\boldsymbol{\cal B}_{\rm FW} \psi \rangle$ is lower bounded if
\begin{align*}
\gamma \leqslant \gamma_{\rm c}^{\cal HSM}={\cal C}_{\alpha}\;\hbar c \beta\,\,.
\end{align*}

\renewcommand{\theequation}{B.\arabic{equation}}
\renewcommand{\thesection}{B}
\setcounter{equation}{0}

\section*{Appendix B: Bessel potential}
\hspace*{\parindent}
For the reader's possible interest, in this appendix, we have gathered some facts about the
spaces of Bessel potentials. Let us recall that the operator $(-\Delta+{\un})^{-s/2}$, for $s > 0$,
is called {\em Bessel potential operator}. Thus, given any $s > 0$, the Bessel potencial $G_s$ is
defined to be that function whose Fourier transform $\widehat{G}_s$ is given by
\[
\widehat{G}_s(\boldsymbol{p})=\frac{1}{(1+|\boldsymbol{p}|^2)^{s/2}}\,\,,
\quad \text{with $\boldsymbol{p} \in \oR^d$}\,\,.
\]

The following is a simple proof of this result for the $K_0$-potential in $\oR^2$. The proof is based on the spherical
symmetry of the $K_0$-function.

\begin{proposition}
Given that $K_0$-potential is spherically symmetric, then
\begin{align*}
\bigl[{\mathscr F}K_0\bigr](\boldsymbol{p})=\widehat{K}_0(\boldsymbol{p})
=\frac{2\pi \hbar^2}{|\boldsymbol{p}|^2+\hbar^2 \beta^2}\,\,.
\end{align*}
\label{FTK0}
\end{proposition}

\begin{proof}
Using our Fourier transform convention displayed in Eq.(\ref{FouTrans}), we started by writing
\begin{align*}
\widehat{K}_0(\boldsymbol{p})
=\int_{-\infty}^\infty \int_{-\infty}^\infty K_0\bigl(\beta \sqrt{x_1^2+x_2^2}\bigr) 
e^{i\hbar^{-1}(p_1 x_1+p_2 x_2)} dx_1dx_2\,\,.
\end{align*}
For a potential depending only upon $r$ (central force field) it is expedient to introduce polar coordinates 
by the formulae $x_1=r \cos \theta$, $x_2=r \sin \theta$, with $r=|\boldsymbol{x}|$, and
similarly in the momentum domain by the formulae $p_1=p \cos \varphi$, $p_2=p \sin \varphi$, with
$p=|\boldsymbol{p}|$. It then follows that the Fourier transform in $d=2$ can be written as
\begin{align}
\widehat{K}_0(\boldsymbol{p})&=\int_0^\infty \int_{-\pi}^{\pi} K_0(\beta r)
e^{i\hbar^{-1}pr \cos(\varphi-\theta)} rdrd\theta \nonumber \\[3mm]
&=\int_0^\infty K_0(\beta r) rdr
\int_{-\pi}^{\pi} e^{i\hbar^{-1}pr \cos(\varphi-\theta)} d\theta\,\,.
\label{FTK01}
\end{align}
Using the integral definition of the zeroth-order Bessel function,
\begin{align*}
J_0(x)=\frac{1}{2\pi} \int_{-\pi}^{\pi} e^{ix \cos(\varphi-\theta)} d\theta
=\frac{1}{2\pi} \int_{-\pi}^{\pi} e^{ix \cos \eta} d\eta\,\,,
\end{align*}
Eq.(\ref{FTK01}) can then be written as
\begin{align*}
\widehat{K}_0(\boldsymbol{p})=2\pi \int_0^\infty K_0(\beta r) J_0(\hbar^{-1}pr)~r\,dr\,\,.
\end{align*}
In this way, the conclusion follows from the Table of Integrals of Gradshtein-Ryzhik~\cite[{\bf 6.521}, $2.^{10}$, p.665]{GR}.
\end{proof}

We quote below without proof the basic properties of $G_s$ relevant to our development 
(more details can be found in Refs.~\cite{Aron,Elias,Ziemer,Adams}): 

$(i)$ if $s > 0$, then $G_s$ is a positive function in $L_1(\oR^d)$ which is analytic except at $0$ and is given by
\begin{align}
G_s(\boldsymbol{x})=\frac{1}{2^{\frac{d+s-2}{2}}\pi^{d/2}\Gamma(s/2)}
K_{\frac{d-s}{2}}(|\boldsymbol{x}|)\;|\boldsymbol{x}|^{\frac{s-d}{2}}\,\,,
\label{IntRepBF0}
\end{align}
where $K_{\frac{d-s}{2}}$ is the modified Bessel function of the second kind also called Bessel-MacDonald function and
$\Gamma$  denotes the Gamma function. The Bessel kernel can also be represented for 
$\boldsymbol{x} \in \oR^d \setminus \{0\}$ by the integral formula
\begin{align}
G_s(\boldsymbol{x})={\frac {1}{(4\pi )^{s/2}\Gamma(s/2)}}
\int _{0}^{\infty } e^{-{\frac {\pi |\boldsymbol{x}|^{2}}{\eta}}}
e^{-\frac{\eta}{4\pi}}~\eta^{-\left(1+{\frac {d-s}{2}}\right)}\,d\eta\,\,.
\label{IntRepBF}
\end{align}

$(ii)$ $G_s * G_\tau=G_{s+\tau}$ if $\tau > 0$.

$(iii)$ as $|\boldsymbol{x}| \to 0$,
\begin{align*}
G_{s}(\boldsymbol{x}) \simeq
\begin{cases}
\displaystyle{{\frac{\Gamma((d-s)/2)}{2^{s}\pi^{d/2}\Gamma(s/2)}}}\;|\boldsymbol{x}|^{s-d} \quad {\text{if}} \quad 0 < s < d\,\,, \\[5mm]
\displaystyle{{\frac{1}{2^{d-1}\pi^{d/2}\Gamma(d/2)}} \ln{\frac{1}{|\boldsymbol{x}|}}} \quad {\text{if}} \quad s=d\,\,, \\[5mm]
\displaystyle{{\frac{\Gamma((d-s)/2)}{2^{d}\pi^{d/2}\Gamma(s/2)}}} \quad {\text{if}} \quad s > d\,\,.
\end{cases}
\end{align*}

$(iv)$ as $|\boldsymbol{x}| \to \infty$,
\begin{align*}
G_{s}(\boldsymbol{x})
\simeq \frac{1}{2^{\frac{d+s-1}{2}}\pi^{\frac{d-1}{2}}\Gamma(s/2)}\;|\boldsymbol{x}|^{\frac{s-d-1}{2}}\;e^{-|\boldsymbol{x}|}\,\,.
\end{align*}

$(v)$ there exists $c > 0$ such that for all $\boldsymbol{x} \in \oR^d$ and all $s \in (0,d)$
\begin{align*}
G_{s}(\boldsymbol{x})
\simeq |\boldsymbol{x}|^{s-d}\;e^{-c|\boldsymbol{x}|}\,\,.
\end{align*}

Closely related to the operator $(-\Delta+{\un})^{-s/2}$ is the {\em Riesz potential operator}, $(-\Delta)^{-s/2}$, 
which has an integral convolution kernel of the form
\begin{align*}
G_{s}(\boldsymbol{x})=\frac{\Gamma((d-s)/2)}{2^s \pi^{d/2}\Gamma(s/2)}\;|\boldsymbol{x}|^{s-d}\,\,, 
\quad \text{if} \quad 0 < s < d\,\,.
\end{align*}

The Bessel potential is a potential similar to the Riesz potential but with better decay properties at infinity.
Comparatively, the Yukawa potential is a particular case of a Bessel potential for $s=2$ in $d=3$,
while the Coulomb potential is an example of a Riesz potential also in $d=3$. Note that according to properties 
$(iii)$ and $(iv)$, for $s=d=2$, the $K_0$-potential behaves as if it were the Yukawa 
potential in $d=2$ ({\em cf.}~\cite[p.1006, Eq.(2.21a)]{Peter}). In~\cite{CoTa}, Cotsiolis-Tavoularis obtained the
best best constants for the inequalities with Riesz, Bessel and Yukawa potential operators.



\end{document}